%% file: main.tex
\documentclass[a4paper,UKenglish,cleveref, autoref, thm-restate]{lipics-v2021}

\title{Obtaining Approximately Optimal and Diverse Solutions via Dispersion}

\author{Jie Gao}{Department of Computer Science, Rutgers University, USA}{jg1555@rutgers.edu}{https://orcid.org/0000-0001-5083-6082}{This work is supported by NSF OAC-1939459, CCF-2118953 and CCF-1934924.}
\author{Mayank Goswami}{Department of Computer Science, Queens College, City University of New York, USA}{mayank.goswami@qc.cuny.edu}{[orcid]}{This work is supported by US National Science Foundation (NSF) awards CRII-1755791 and CCF-1910873.}
\author{Karthik C.\ S.}{Department of Computer Science, Rutgers University, USA}{karthik.cs@rutgers.edu}{}{This work was supported by Subhash Khot's Simons Investigator Award and by a grant from the Simons Foundation, Grant Number 825876, Awardee Thu D. Nguyen.}
\author{Meng-Tsung Tsai}{Institute of Information Science, Academia Sinica, Taiwan}{mttsai@iis.sinica.edu.tw}{https://orcid.org/0000-0002-2243-8666}{
This research was supported in part by the Ministry of Science and Technology of Taiwan under contract MOST grant 109-2221-E-001-025-MY3.
}
\author{Shih-Yu Tsai}{Department of Computer Science, Stony Brook University, USA}{shitsai@cs.stonybrook.edu}{}{}
\author{Hao-Tsung Yang}{School of Informatics, University of Edinburgh, UK}{yhaotsu@exseed.ed.ac.uk}{}{}

\authorrunning{Gao, Goswami, C.S., Tsai, Tsai, and Yang} 

\Copyright{Jie Gao, Mayank Goswami, Karthik C.\ S., Meng-Tsung Tsai, Shih-Yu Tsai, and Hao-Tsung Yang} 



\category{} 

\relatedversion{} 

\supplement{}

\acknowledgements{}

\nolinenumbers 

\hideLIPIcs  

\EventEditors{}
\EventNoEds{2}
\EventLongTitle{}
\EventShortTitle{}
\EventAcronym{}
\EventYear{}
\EventDate{}
\EventLocation{}
\EventLogo{}
\SeriesVolume{}
\ArticleNo{1}

\usepackage{amssymb}
\usepackage{graphicx}
\usepackage{microtype}
\usepackage{bbm}
\usepackage{amsmath}
\usepackage{amsthm}
\usepackage{enumitem}
\usepackage{amsfonts}
\usepackage{enumerate}
\usepackage{wrapfig}
\usepackage{multirow}
\usepackage{nicefrac}
\usepackage[colorinlistoftodos,prependcaption,textsize=tiny]{todonotes}

\definecolor{forestgreen}{rgb}{0.13, 0.55, 0.13}

\usepackage{algpseudocode}

\usepackage{amsfonts,amsmath}

\algnewcommand\algorithmicforeach{\textbf{for each}}
\algdef{S}[FOR]{ForEach}[1]{\algorithmicforeach\ #1\ \algorithmicdo}
\newcommand{\SOL}{\mathsf{Sol}}

\newcommand{\m}{\Delta}
\newcommand{\goal}{\mathsf{goal}}
\newcommand{\1}{\mathbbm{1}}

\date{}

\begin{document}
	
	\maketitle

\input{abstract}

\input{intro}
\input{related_work}
\input{Problem_Statement}

\input{Reduction_to_BCO}
\input{Application1-3}
\input{Application4}

\section{Conclusion}\label{sec:conclusion}
 
 We obtained bi-approximation algorithms for the diversity computational problems associated with approximate minimum-weight bases of a matroid, shortest paths, matchings, and spanning trees. Our general reduction lemma can be applied to any computational problem for which the DUALRESTRICT version defined by Papadimitriou and Yannakakis \cite{papadimitriou2000approximability} can be solved in polynomial time. There are (at least) three open questions:
 \begin{itemize}
    \item Are there other interesting examples of problems to which our reduction lemma applies?
     \item It is noteworthy that while Types 3, 4 and 5 in Theorem~\ref{thm:reduction} contain the class for which DUALRESTRICT can be solved efficiently, the case of $a=b=1+\epsilon$ missing in our theorem would include the larger GAP class studied by Papadimitriou and Yannakakis \cite{papadimitriou2000approximability}. Our current techniques to prove Theorem~\ref{thm:reduction} do not allow us to translate such bi-approximations of BCOs into those for DCPs (at least not without demanding certain stronger condition than in Type 5).
     \item Can the diversity computational problem for approximate TSP tours or Max-TSP tours be solved in polynomial time?
 \end{itemize}

\clearpage

	\bibliography{ref}
	
\end{document}

%% file: abstract.tex
\begin{abstract}

There has been a long-standing interest in computing diverse solutions to optimization problems. Motivated by reallocation of governmental institutions in Sweden, in 1995 J. Krarup \cite{krarup1995peripatetic} posed the problem of finding $k$ edge-disjoint Hamiltonian Circuits of minimum total weight, called the peripatetic salesman problem (PSP). Since then researchers have investigated the complexity of finding diverse solutions to spanning trees, paths, vertex covers, matchings, and more. Unlike the PSP that has a constraint on the total weight of the solutions, recent work has involved finding diverse solutions that are all optimal.

However, sometimes the space of exact solutions may be too small to achieve sufficient diversity. Motivated by this, we initiate the study of obtaining sufficiently-diverse, yet approximately-optimal solutions to optimization problems. Formally, given an integer $k$, an approximation factor $c$, and an instance $I$ of an optimization problem, we aim to obtain a set of $k$ solutions to $I$ that a) are all $c$ approximately-optimal for $I$ and b) maximize the diversity of the $k$ solutions. Finding such solutions, therefore, requires a better understanding of the global landscape of the optimization function.  

We show that, given any metric on the space of solutions, and the diversity measure as the sum of pairwise distances between solutions, this problem can be solved by combining ideas from dispersion and multicriteria optimization. We first provide a general reduction to an associated budget-constrained optimization (BCO) problem, where one objective function is to be maximized (minimized) subject to a bound on the second objective function. We then prove that bi-approximations to the BCO can be used to give bi-approximations to the diverse approximately optimal solutions problem with a little overhead. We then give the following applications of our result:

\begin{itemize}
    \item Polynomial time $(2,c)$ bi-approximation algorithms for diverse $c$-approximate maximum matchings, and for diverse $c$-approximate $st$ shortest paths.
    \item Polynomial time $(4,2c)$ bi-approximation, a PTAS $(4,(1+\epsilon)c)$ bi-approximation and pseudopolynomial $(4,c)$ bi-approximation algorithms for diverse $c$-approximate minimum weight bases of a matroid. In particular, this gives us diverse $c$-approximate minimum spanning trees, advancing a step towards achieving diverse $c$-approximate TSP tours.
\end{itemize}

We explore the connection to the field of multiobjective optimization and show that the class of problems to which our result applies includes those for which the associated DUALRESTRICT problem defined by Papadimitriou and Yannakakis \cite{papadimitriou2000approximability}, and recently explored by Herzel et al. \cite{herzel2021one} can be solved in polynomial time.

On the way we prove a conditional 2-hardness of approximation lower bound on the problem of $k$-dispersion in metric spaces, resolving (up to ETH) the optimality of the greedy farthest point heuristic for dispersion. This also provides evidence of tightness of the factor 2 for diversity in most of our algorithms.

\end{abstract}

%% file: intro.tex
\section{Introduction}\label{sec:intro}

Techniques for optimization problems focus on obtaining optimal solutions to an objective function and have widespread applications ranging from machine learning, operations research, computational biology, networks, to geophysics, economics, and finance. However, in many scenarios, the optimal solution is not only computationally difficult to obtain, but can also render the system built upon its utilization vulnerable to adversarial attacks. Consider a patrolling agent tasked with monitoring $n$ sites in the plane. The most efficient solution (i.e., maximizing the frequency of visiting each of the $n$ sites) would naturally be to patrol along the tour of shortest length\footnote{We assume without loss of generality that the optimal TSP is combinatorially unique by a slight perturbation of the distances.} (the solution to TSP - the Traveling Salesman Problem). However, an adversary who wants to avoid the patroller can also compute the shortest TSP tour and can design its actions strategically~\cite{Yang2019-uc}. Similarly, applications utilizing the minimum spanning tree (MST) on a communication network may be affected if an adversary gains knowledge of the network~\cite{commander2007wireless}; systems using solutions to a linear program (LP) would be vulnerable if an adversary gains knowledge of the program's function and constraints. 

One way to address the vulnerability is to use a set of approximately optimal solutions and randomize among them. However, this may not help much to mitigate the problem, if these approximate solutions are combinatorially too ``similar'' to the optimal solution. For example, all points in a sufficiently small neighborhood of the optimal solution on the LP polytope will be approximately optimal, but these solutions are not too much different and the adversaries can still effectively carry out their attacks. Similarly one may use another tree instead of the MST, but if the new tree shares many edges with the MST the same vulnerability persists. Thus $k$-best enumeration algorithms (\cite{gabow1977two, hara2017enumerate,Lawler72,lindgren2017exact,murty1968algorithm}) developed for a variety of problems fall short in this regard.

One of the oldest known formulations is the Peripatetic Salesman problem (PSP) by Krarup \cite{krarup1995peripatetic}, which asks for $k$ edge-disjoint Hamiltonian circuits of minimum total weight in a network. Since then, several researchers have tried to  compute diverse solutions for several optimization problems \cite{baste2022diversity,baste2019fpt,fomin2020diverse,hanaka2021finding}. Most of these works are on graph problems, and diversity usually corresponds to the size of the symmetric difference of the edge sets in the solutions. Crucially, almost all of the aforementioned work demands either every solution individually be optimal, or the set of solutions in totality (as in the case of the PSP) be optimal. Nevertheless, \textbf{the space of optimal solutions may be too small to achieve sufficient diversity}, and it may just be singular (unique solution). In addition, for NP-complete problems finding just one optimal solution is already difficult. While there is some research that takes the route of developing FPT algorithms for this setting \cite{baste2019fpt,fomin2021diverse}, to us it seems practical to also consider the relaxation to approximately-optimal solutions. 

This motivates the problem of finding a set of diverse and \textit{approximately optimal} solutions, which is the problem considered in this article. The number of solutions $k$ and the desired approximation factor $c>1$ is provided by the user as input. Working in the larger class gives one more hope of finding diverse solutions, yet every solution has a guarantee on its quality.

\subsection{Our Contributions}

We develop approximation algorithms for finding $k$ solutions to the given optimization problem: for every solution, the quality is bounded by a user-given approximation ratio $c>1$ to the optimal solution and the diversity of these $k$ solutions is maximized. Given a metric on the space of solutions to the problem, we consider the diversity measure given by the sum (or average) of pairwise distances between the $k$ solutions. Combining ideas from the well-studied problem on dispersion (which we describe next), we reduce the above problem to a budget constrained optimization (BCO) program. 

\subsection{Dispersion} Generally speaking, if the optimization problem itself is $\mathcal{NP}$-hard, finding diverse solutions for that problem is also $\mathcal{NP}$-hard (see Proposition~\ref{prop:general_NP} for more detail). On the other hand, interestingly, even if the original problem is not $\mathcal{NP}$-hard, finding diverse and approximately optimal solutions can still be $\mathcal{NP}$-hard. 
This is due to the connection of the diversity maximization objective with the general family of problems that consider selecting $k$ elements from the given input set with maximum ``dispersion'', defined as max-min distance, max-average distance, and so on. 

The dispersion problem has a long history, with many variants both in the metric setting and the geometric setting~\cite{ERKUT199048,Kuby87dispersion,WANG1988281}. 
For example, finding a subset of size $k$ from an input set of $n$ points in a metric space that maximizes the distance between closest pairs or the sum of distances of the $k$ selected points are both $\mathcal{NP}$-hard~\cite{abbar2013kdnn,ravi1994heuristic}. 
For the max-sum dispersion problem, the best known approximation factor is 2 in the metric setting~\cite{birnbaum09improved,Hassin1997-uk} and $1.571+\epsilon$ with $\epsilon>0$ for the 2D Euclidean setting~\cite{ravi1994heuristic}. It is not known whether there is a PTAS for the max-min dispersion problem. 

\noindent\textbf{Conditional hardness} In this paper, we provide new hardness results for the max-sum $k$-dispersion problem. We show that, assuming a standard complexity theoretic conjecture (Exponential Time Hypothesis), there is no polynomial time algorithm for the $k$-dispersion problem in a metric space which approximates it to a factor better than $2$ (Theorem~\ref{thm:dispersion}). This provides a tight understanding for the max-sum $k$-dispersion problem which has been studied since the 90s \cite{ravi1994heuristic}.  

\noindent\textbf{Dispersion in exponentially-sized space} We make use of the general framework of the 2-approximation algorithm~\cite{BorodinJLY17,ravi1994heuristic} to the max-sum $k$-dispersion problem, a greedy algorithm where the $j$th solution is chosen to be the most distant/diverse one from the first $j-1$ solutions.
Notice that in our setting, there is an important additional challenge to understand the space within which the approximate solutions stay. In all of the problems we study, the total number of solutions can be \textit{exponential in the input size}. Thus we need to have a non-trivial way of navigating within this large space and carry furthest insertion without considering all points in the space. This is where our reduction to budget constrained problem comes in.

\noindent\textbf{Self avoiding dispersion} Furthermore, even after implicitly defining the $j$th furthest point insertion via some optimization problem, one needs to take care that the (farthest, in terms of sum of distances) solution does not turn out to equal one of the previously found $j-1$ solutions, as this is a requirement for the furthest point insertion algorithm. This is an issue one faces because of the implicit nature of the furthest point procedure in the exponential-sized space of solutions: in the metric $k$-dispersion problem, it was easy to guarantee distinctness as one only considered the $n-(j-1)$ points not yet selected. 

\subsection{Reduction to Budget Constrained Optimization} Combining with dispersion, we reduce the diversity computational problem to a budget constrained optimization (BCO) problem where the budget is an upper (resp. lower) bound if the quality of solution is described by a minimization (resp. maximization) problem. Recall that the number of solutions $k$ and the approximation factor $c$ is input by the user.

We show how using an $(a,b)$ bi-approximation algorithm for the BCO problem provides a set of $O(a)$-diverse, $bc$ approximately-optimal solutions to the diversity computational problem (the hidden constant is at most $4$). This main reduction is described in~\cref{thm:reduction}.

The main challenge in transferring the bi-approximation results because of a technicality that we describe next. Let $\Omega(c)$ be the space of $c$ approximate solutions. A $(*, b)$ bi-approximation algorithm to the BCO relaxes the budget constraint by a factor $b$, and hence only promises to return a faraway point in the larger space $\Omega(b\cdot c)$. Thus bi-approximation of BCO do not simply give a farthest point insertion in the space of solutions, but return a point in a larger space. Nevertheless, we prove in Lemma~\ref{largermetric} that in most cases, one loses a factor of at most $4$ in the approximation factor for the diversity.

Once the reduction to BCOs is complete, for diverse approximate matchings, spanning trees and shortest paths we exploit the special characteristics of the corresponding BCO to solve it optimally ($a=b=1$). For other problems such as diverse approximate minimum weight spanning trees, and the more general minimum weight bases of a matroid, we utilize known bi-approximations to the BCO to obtain bi-approximations for the diversity problem. For all problems except diverse (unweighted) spanning trees\footnote{While an exact algorithm for diverse unweighted spanning trees is known \cite{hanaka2021finding}, we give a faster (by a factor $\Omega(n^{1.5}k^{1.5}/\alpha(n, m))$ where $\alpha(\cdot)$ denotes the inverse of the Ackermann function) $2$-approximation here.}, our algorithms are the first polynomial time bi-approximations for these problems.

We also connect to the wide literature on multicriteria optimization and show that our result applies to the entire class of problems for which the associated DUALRESTRICT problem (defined by Papadimitriou  and Yannakakis \cite{papadimitriou2000approximability}, and recently studied by Herzel et al. \cite{herzel2021one}) has a polynomial time solution. We discuss this in more detail after presenting our reduction.

 The rest of this paper is organized as follows: we survey related work in~\cref{sec:relatedwork}, and formulate the problem in~\cref{sec:framework}. In~\cref{sec:reduction} we describe our main results on dispersion (Theorem~\ref{thm:dispersion}) and the reduction to the budget constrained optimization problem (Theorem~\ref{thm:reduction}). \cref{sec:st,sec:shortestpath,sec:matching,sec:matroid} describe four applications of our technique to various problems such as diverse approximate matchings, shortest paths, minimum spanning trees, and minimum weight bases of a matroid. We remark that this list is by no means exhaustive, and we leave finding other interesting optimization problems which are amenable to our approach for future research. We end with open problems in~\cref{sec:conclusion}.

%% file: related_work.tex
\section{Related Work}\label{sec:relatedwork}

Recently there has been a surge of interest in the tractability of finding diverse solutions for a number of combinatorial optimization problems, such as spanning trees, minimum spanning trees, $k$-paths, shortest paths, $k$-matchings,  etc.~\cite{fomin2020diverse,fomin2021diverse,Hanaka22,Hanaka21,hanaka2021finding}. Most of the existing work focuses on finding diverse optimal solutions. In cases when finding the optimal solution is NP-complete, several works have focused on developing FPT algorithms \cite{baste2019fpt,fomin2021diverse}. Nevertheless, as pointed out in~\cite{Hanaka21}, it would be more practical to consider finding a set of diverse ``short'' paths rather than one set of diverse shortest paths. They show that finding a set of approximately shortest paths with the maximum diversity is NP-hard, but leave the question of developing approximation algorithms open, a question that we answer in our paper for several problems. Similarly the problem of finding diverse maximum matchings was proved to be NP-hard in \cite{fomin2020diverse}. We remark that the main difference between our result and previous work is that our algorithms can find a diverse set of $c$-approximate solutions in polynomial time. If the attained diversity is not sufficient for the application, the user can input a larger $c$, in hopes of increasing it.

\noindent\textbf{Multicriteria Optimization:} In this domain, several optimization functions are given on a space of solutions. Clearly, there may not be a single solution that is the best for all objective functions, and researchers have focused on obtaining Pareto-optimal solutions, which are solutions that are non-dominated by other solutions. Put differently, a solution is Pareto-optimal if no other solution can have a better cost for all criteria. Since exact solutions are hard to find, research has focused on finding $\epsilon$ Pareto-optimal solutions, which are a $1+\epsilon$ factor approximations of Pareto-optimal solutions. Papadimitriou and Yannakakis \cite{papadimitriou2000approximability} showed that under pretty mild conditions, any mutlicriteria optimization problem admits an $\epsilon$ Pareto-optimal set of fully polynomial cardinality. In terms of being able to \textit{find} such an $\epsilon$ Pareto-optimal  set, they show that a (FPTAS) PTAS exists for the problem if and only if an associated GAP problem can be solved in (fully) polynomial time. Very recently, Herzel et al.\cite{herzel2021one} study the class of problems for which a FPTAS or PTAS exists for finding $\epsilon$ Pareto-optimal solutions that are \textit{exact} in one of the criteria. Clearly such problems are a subset of the ones characterized by GAP. Herzel et al. \cite{herzel2021one} characterize the condition similarly: a FPTAS (PTAS) exists if and only if an associated DUALRESTRICT problem can be solved in (fully) polynomial  time. For more details we refer  the reader to the survey by Herzel at al. \cite{herzel2021approximation}.

%% file: Problem_Statement.tex
\section{Problem Statement}\label{sec:framework}

First, we define some notations.
We use the definition of optimization problems given in \cite{ACGKMP99} with additional formalism as introduced in \cite{GK20}.

\begin{definition}[Optimization Problem]
An \textbf{optimization problem} $\Pi$ is characterized by the following quadruple of objects $(I_\Pi,\SOL_{\Pi},\m_\Pi,\goal_\Pi)$, where:
\begin{itemize}
\item $I_\Pi$ is the set of instances of $\Pi$. In particular for every $d\in\mathbb{N}$, $I_\Pi(d)$ is the set of instances of $\Pi$ of input size at most $d$ (bits); 
\item $\SOL_\Pi$ is a function that associates to any input instance $x\in I_\Pi$ the set of feasible solutions of $x$;
\item $\m_\Pi$ is the measure  function\footnote{We define the measure function only for feasible solutions of an instance. Indeed if an algorithm solving the optimization problem outputs a non-feasible solution then, the measure just evaluates to -1 in case of maximization problems and $\infty$ in case of minimization problems.}, defined for pairs $(x,y)$ such that $x\in I_\Pi$ and $y\in \SOL_\Pi(x)$. For every such pair $(x,y)$,  $\m_\Pi(x,y)$ provides a non-negative integer which is the value of the feasible solution $y$; 
\item $\goal_\Pi\in\{\min,\max\}$ specifies whether $\Pi$ is a maximization or minimization problem.
\end{itemize}
\end{definition}

We would like to identify a subset of our solution space which are (approximately) optimal with respect to our measure function. To this effect, we define a notion of approximately optimal feasible solution.

\begin{definition}[Approximately Optimal Feasible Solution]
Let $\Pi(I_\Pi,\SOL_{\Pi},\m_\Pi,\goal_\Pi)$ be an optimization problem and let $c\ge 1$. For every $x\in I_\Pi$ and $y\in \SOL_\Pi(x)$ we say that $y$ is a \textbf{$c$-approximate optimal solution of} $x$ if for every $y'\in \SOL_\Pi(x)$ we have $\m_\Pi(x,y)\cdot c \ge \m_\Pi(x,y')$ if $\goal_\Pi=\max$ and  $\m_\Pi(x,y)\le \m_\Pi(x,y')\cdot c$ if $\goal_\Pi=\min$.
\end{definition}

\begin{definition}[Computational Problem]
Let $\Pi(I_\Pi,\SOL_{\Pi},\m_\Pi,\goal_\Pi)$ be an optimization problem and let $\lambda:\mathbb{N}\to\mathbb{N}$. The \textbf{computational problem associated with $(\Pi,\lambda)$} is given as input an instance $x\in I_\Pi(d)$ (for some $d\in\mathbb{N}$) and real $c:=\lambda(d)\ge 1$ find a $c$-approximate optimal feasible solution of $x$. 
\end{definition}

\begin{definition}[Diversity Computational Problem]\label{defproblem}
Let $\Pi(I_\Pi,\SOL_{\Pi},\m_\Pi,\goal_\Pi)$ be an optimization problem and let $\lambda:\mathbb{N}\to\mathbb{N}$. Let $\sigma_{\Pi,t}$ be a diversity measure that maps every $t$ feasible solutions of an instance of $I_\Pi$ to a non-negative real number. The \textbf{diversity computational problem associated with $(\Pi,\sigma_{\Pi,t},k, \lambda)$} is given as input an instance $x\in I_\Pi(d)$ (for some $d\in\mathbb{N}$), an integer $k:=k(d)$, and real $c:=\lambda(d)\ge 1$, find $k$-many $c$-approximate solutions $y_1,\ldots ,y_k$ to $x$ which maximize the value of $\sigma_{\Pi,k}(x,y_1,\ldots ,y_k)$. 
\end{definition}

\begin{proposition}\label{prop:general_NP}
Let $\Pi(I_\Pi,\SOL_{\Pi},\m_\Pi,\goal_\Pi)$ be an optimization problem and let $\lambda:\mathbb{N}\to\mathbb{N}$. Let $\sigma_{\Pi,t}$ be a diversity measure that maps every $t$ feasible solutions of an instance of $I_\Pi$ to a non-negative real number.
If the computational problem associated with $(\Pi,\lambda)$ is $\mathcal{NP}$-hard then, the diversity computational problem associated with $(\Pi,\sigma_{\Pi,t},\lambda)$ is also $\mathcal{NP}$-hard. 
\end{proposition}

Therefore the interesting questions arise when we compute problems associated with $(\Pi,\lambda)$ which are in $\mathcal{P}$, or even more when,  $(\Pi,\1)$ is in $\mathcal{P}$
where $\1$ is the constant function which maps every element of the domain to 1. For the remainder of this paper, we will consider $\lambda(d)$ to be the constant function, and will simply refer to the constant as $c$.

Finally, we define bicriteria approximations for the diversity computational problem:

\begin{definition}[$(\alpha,\beta)$ bi-approximation for the Diversity Computational Problem]\label{defbiapprox}
Consider the diversity computational problem associated with $(\Pi,\sigma_{\Pi,t}, k, c)$, and a given instance $x \in I_{\Pi}(d)$ (for some $d \in \mathbb{N}$). An algorithm is called an $(\alpha,\beta)$ bi-approximation for the diversity computational problem if it outputs $k$ feasible solutions $y_1,\ldots,y_k$ such that a) $y_i$ is a $\beta \cdot c$-approximate optimal feasible solution to $x$ for all $ 1 \leq i \leq k$, and b) for any set $y_{1}^{'},\ldots,y_{k}^{'}$ of $k$-many $c$-approximate optimal feasible solutions, $\sigma_{\Pi,k}(y_1,\cdots,y_k)\cdot \alpha \geq \sigma_{\Pi,k}(y_{1}^{'},\cdots,y_{k}^{'})$. Furthermore, such an algorithm is said to run in polynomial time if the running time is polynomial in $d$ and $k$. 
\end{definition}

%% file: Reduction_to_BCO.tex
\section{The Reduction: Enter Dispersion and Biobjective Optimization}\label{sec:reduction}

As stated in the introduction, our problems are related to the classical dispersion problem in a metric space. Here we state the dispersion problem and show that under the exponential time hypothesis, 2-approximation is actually tight for the $k$-dispersion problem. We will then use dispersion to reduce the problem of finding diverse, approximately optimal solutions to solving an associated budget constrained optimization problems.

\subsection{Dispersion Problem}

\begin{definition}[$k$-Dispersion, total distance]\label{def:dispersion}
Given a finite set of points $P$ whose pairwise distances satisfy the triangle inequality and an integer $k \ge 2$, find a set $S \subseteq P$ of cardinality $k$ so that $W(S)$ is maximized, where $W(S)$ is the sum of the pairwise distances between points in $S$. 
\end{definition}

The main previous work on the $k$-dispersion problem relevant to us is~\cite{ravi1994heuristic}, where the problem was named as Maximum-Average Facility Dispersion problem with triangle inequality (MAFD-TI). The problems are equivalent as maximizing the average distance between the points also maximizes the sum of pairwise distances between them and vice-versa. 

The $k$-dispersion problem is $\mathcal{NP}$-hard, but one can find a set $S$ whose $W(S)$ is at least a constant factor of the maximum possible in polynomial time by a greedy procedure~\cite{ravi1994heuristic}. We call the greedy procedure \emph{furthest insertion}. It works as follows. Initially, let $P$ be a singleton set that contains an arbitrary point from the given set. While $|S| < k$, add to $
S$ a point $x \notin S$ so that $W(S \cup \{x\}) \ge W(S \cup \{y\})$ for any $y \notin S$. Repeat the greedy addition until $S$ has size $k$. The final $S$ is a desired solution, which is shown to be a 4-approximation in~\cite{ravi1994heuristic}. It is worth noting that the furthest insertion in~\cite{ravi1994heuristic}  initializes $S$ as a furthest pair of points in the given set, and the above change does not worsen the approximation factor. 
In a later paper~\cite{BorodinJLY17}, the above greedy algorithm that chooses an arbitrary initial point is shown to be a 2-approximation,
which is a tight bound for this algorithm~\cite{birnbaum09improved}.

\begin{lemma}[Furthest Insertion in~\cite{BorodinJLY17,ravi1994heuristic}]\label{4approx}
The $k$-dispersion problem can be 2-approximated by the furthest insertion algorithm.
\end{lemma}

The running time of the furthest insertion algorithm is polynomial in $|S|$ the size of $S$, as it performs $k$ iterations, each performing at most $O(k|S|)$ distance computations/lookups. Note that in our case $S$ is the collection of objects of a certain type (matchings, paths, trees, etc.). Hence the size of our metric space is typically exponential in $|V|$ and $|E|$. This adds a new dimension of complexity to the traditional dispersion problems studied.

We prove that the $k$-dispersion problem cannot be approximated better than $2$ given the Exponential Time Hypothesis. 

\begin{theorem}[Inapproximability of $k$-Dispersion]\label{thm:dispersion}
Assuming the Exponential Time Hypothesis, for every $\varepsilon>0$, no algorithm running in polynomial time can approximate the $k$-Dispersion problem to a factor of $2-\varepsilon$ .
\end{theorem}
\begin{proof}
We use the inapproximability of the densest $k$ subgraph problem.

\begin{theorem}[Inapproximability of Densest $k$-subgraph with Perfect Completeness \cite{M17}]\label{thm:denseksubgraph}
Given the Exponential Time Hypothesis, for every $\delta>0$, no algorithm running in polynomial time can distinguish between the following two cases, given as input a graph $G$ and an integer $k$:
\begin{description}
\item[Completeness:] There is a clique of size $k$ in $G$.
\item[Soundness:] The subgraph induced by any $k$ vertices in $G$ has at most $\delta\cdot \binom{k}{2}$ many edges.
\end{description}
\end{theorem}

Given a graph $G=(V,E)$, we define a distance function $\Delta:V\times V\to\mathbb{R}^{\ge 0}$ as follows. 
\[
\forall u,v\in V,\ \Delta(u,v)=\begin{cases}
0\text{ if }u=v,\\
2\text{ if }(u,v)\in E,\\
1\text{ otherwise}.
\end{cases}
\]
Note that $(V,\Delta)$ is a metric space. Suppose there is a set of $k$ vertices $S$ such that it forms a clique in $G$, then, $\sum_{\{u,v\}\in \binom{S}{2}}\Delta(u,v)=2 \cdot \binom{k}{2}.$
On the other hand, if for some set $S$ of $k$ vertices we have that the subgraph induced by $S$ in $G$ has at most $\delta\cdot \binom{k}{2}$ many edges, then,
$
\sum_{\{u,v\}\in \binom{S}{2}}\Delta(u,v)\le (1+\delta)\cdot \binom{k}{2}.
$
The theorem statement thus follows from Theorem~\ref{thm:denseksubgraph} by setting $\delta:=\varepsilon/2$.
\end{proof}

\subsection{Reduction to Budget Constrained Optimization}

Recall the definitions of the Diversity Computational Problem (Definition~\ref{defproblem}) and $(a,b)$ bi-approximations (Definition~\ref{defbiapprox}). As the input instance $x \in I_{\Pi}$ will be clear from context, we drop the dependence on $x$, and assume a fixed input instance to a computational problem. Thus $\SOL_{\Pi}$ will denote the set of feasible solutions, and $\Delta_{\Pi}(y)$ the measure of the feasible solution $y$.

\noindent\textbf{Diversity and similarity measures from metrics} Let $d: \SOL_{\Pi} \times \SOL_{\Pi} \rightarrow \mathbb{R}^{+}$ be a metric on the space of feasible solutions. When such a metric is available, we will consider the diversity function $\sigma_{\Pi,t}: \SOL_{\Pi} \times \cdots \times \SOL_{\Pi} \rightarrow \mathbb{R}^{+}$ that assigns the diversity measure $\sum_{i,j} d(y_{i},y_{j})$ to a $t$-tuple of feasible solutions $(y_1,\cdots,y_t)$. Also, given such a metric $d$, define $D$ to be the diameter of $\SOL_{\Pi}$ under $d$, i.e., $D =\max_{y,y' \in \SOL_{\Pi}} d(y,y')$. In many cases, we will be interested in the \textit{similarity} measure $s_{\Pi,t}$ defined by $s_{\Pi,t}(y_1,\cdots,y_t) = \sum_{i,j} (D - d(y_{i},y_{j}))$. The examples the reader should keep in mind are graph objects such as spanning trees, matchings, shortest paths, Hamiltonian circuits, etc., such that $d(y,y')$ denotes the Hamming distance, a.k.a. size of the symmetric difference of the edge sets of $y$ and $y'$, and $s$ denotes the size of their intersection.

In the remainder of the paper we consider the above total distance (resp. similarity) diversity measures $\sigma_{\Pi,t}$ arising from the metric $d$ (resp. similarity measure $s$), and we will parameterize the problem by $d$ (resp. $s$) instead.

\begin{definition}[Budget Constrained Optimization]
Given an instance of a computational problem $\Pi$, a $c \geq 1$, and a set $y_1,\ldots,y_i$ of feasible solutions in $\SOL_{\Pi}$, define the metric budget constrained optimization problem $BCO(\Pi, (y_1,\ldots,y_i), c, d)$ as follows:

\begin{itemize}
    \item If $\goal_{\pi} = \min$, define $\Delta^{*}:= \min_{y \in \SOL_{\Pi}} \m_{\Pi}(y)$. Then
    $BCO(\Pi, (y_1,\ldots,y_i), c, d)$ is the problem
    
\begin{equation}
\begin{aligned}
\max \quad & f_{d}(y):= \sum_{j=1}^{i} d(y,y_i)\\
\textrm{s.t.} \quad & \m_\Pi(y) \leq c \cdot \Delta^{*} \\
  & y \in \SOL_{\Pi} \setminus \{y_1,\ldots,y_i\}    \\
\end{aligned}
\end{equation}

\item If $\goal_{\pi} = \max$, define $\Delta^{*}:= \max_{y \in \SOL_{\Pi}} \m_{\Pi}(y)$. Then
    $BCO(\Pi, (y_1,\ldots,y_i), c, d)$ is the problem
    
\begin{equation}
\begin{aligned}
\max \quad & f_{d}(y):= \sum_{j=1}^{i} d(y,y_i)\\
\textrm{s.t.} \quad & \m_\Pi(y) \cdot c \geq  \Delta^{*} \\
  & y \in \SOL_{\Pi} \setminus \{y_1,\ldots,y_i\}    \\
\end{aligned}
\end{equation}

\item Define the similarity budget constrained problem $BCO(\Pi, (y_1,\ldots,y_i), c, s)$, where $s$ is a given similarity measure, with the same constraint set as above (depending on $\goal_{\pi}$), but with the objective function changed to $g_{s}(y):= \min \sum_{j=1}^{i} s(y,y_i)$ instead of $f_{d}(y)=\max \sum_{j=1}^{i} d(y,y_i)$.

\end{itemize}
\end{definition}

\begin{definition}[Bi-approximation to BCO]
An algorithm for an associated BCO is called an $(a,b)$ bi-approximation algorithm if for any $1 \leq i \leq k$, it outputs a solution $y$ such that the following holds.
\begin{itemize}
    \item If $\goal_{\Pi} =\min$ and the associated BCO is $BCO(\Pi, (y_1,\ldots,y_i), c, d)$, then a) $y \in \SOL_{\Pi} \setminus \{y_1,\cdots,y_i\}$, b) $\Delta_{\Pi}(y) \leq b \cdot c \cdot \Delta^{*}$, and c) for all $y'$ satisfying the constraints of $BCO(\Pi, (y_1,\ldots,y_i), c, d)$, $f_{d}(y) \cdot a \geq f_{d}(y')$.
    \item If $\goal_{\Pi} =\max$ and the associated BCO is $BCO(\Pi, (y_1,\ldots,y_i), c, d)$, then a) $y \in \SOL_{\Pi} \setminus \{y_1,\cdots,y_i\}$, b) $\Delta_{\Pi}(y) \cdot b \cdot c \geq \Delta^{*}$, and c) for all $y'$ satisfying the constraints of $BCO(\Pi, (y_1,\ldots,y_i), c, d)$, $f_{d}(y) \cdot a \geq f_{d}(y')$.
    \item If $\goal_{\Pi} =\min$ and the associated BCO is $BCO(\Pi, (y_1,\ldots,y_i), c, s)$, then a) $y \in \SOL_{\Pi} \setminus \{y_1,\cdots,y_i\}$, b) $\Delta_{\Pi}(y) \leq b \cdot c \cdot \Delta^{*}$, and c) for all $y'$ satisfying the constraints of $BCO(\Pi, (y_1,\ldots,y_i), c, s)$, $g_{s}(y) \leq g_{s}(y') \cdot a$.
    \item If $\goal_{\Pi} =\max$ and the associated BCO is $BCO(\Pi, (y_1,\ldots,y_i), c, s)$, then a) $y \in \SOL_{\Pi} \setminus \{y_1,\cdots,y_i\}$, b) $\Delta_{\Pi}(y) \cdot b \cdot c \geq \Delta^{*}$, and c) for all $y'$ satisfying the constraints of $BCO(\Pi, (y_1,\ldots,y_i), c, s)$, $g_{s}(y) \leq g_{s}(y') \cdot a$.
\end{itemize}

\end{definition}

We are now ready to state our main theorem.

\begin{theorem}[Reduction of DCP to BCO]
\label{thm:reduction}
Consider an input $(\Pi, k,d, c)$ to the diversity computational problem (DCP).
\begin{itemize}

    \item For metric BCO,
    \begin{enumerate} 
    \item An $(a,1)$ bi-approximation to $BCO(\Pi, (y_1,\ldots,y_i), c, d)$ can be used to give a $(2a,c)$ approximation to the DCP, and 
    \item An $(a,b)$ bi-approximation to $BCO(\Pi, (y_1,\ldots,y_i), c, d)$ can be used to give a $(4a,bc)$ approximation to the DCP. 
    \end{enumerate}
    \item For similarity BCO,
    \begin{enumerate}[resume]
        \item A $(1,1)$ bi-approximation to $BCO(\Pi, (y_1,\ldots,y_i), c, s)$ can be used to give a $(2,c)$ approximation to the DCP, 
    \item A $(1,b)$ bi-approximation to $BCO(\Pi, (y_1,\ldots,y_i), c, s)$ can be used to give $(4,bc)$ approximation to the DCP, 
    \item A $(1+\epsilon, 1)$  bi-approximation to $BCO(\Pi, (y_1,\ldots,y_i), c, s)$ can be used to give $(4,c)$ approximation to the DCP, under the condition that the average pairwise distance in the optimal solution to the DCP is at least $D \frac{4\epsilon}{1+2\epsilon}$. 
    \end{enumerate}
    
\end{itemize}

In all of the above, the overhead for obtaining a bi-approximation for the DCP, given a bi-approximation for BCO problem, is $O(k)$. 
\end{theorem}

\begin{proof}
Given a $BCO(\Pi, (y_1,\ldots,y_i), c, d)$, denote by $\Omega(c, (y_1,\ldots,y_i))$ its feasible set. When $b=1$, the returned solution $y$ lies in $\Omega(c, (y_1,\ldots,y_i))$. On the other hand when $b >1$, the returned solution lies in the larger space $\Omega(bc, (y_1,\ldots,y_i))$. Our proof will handle these cases separately. We will first state two lemmas concerning the greedy heuristic of farthest point insertion, the first of which we suspect is folklore, but we could not find a reference. The second lemma is the heart of the argument, and is (as far as we know) novel. 

\begin{lemma}\label{farthestapprox}
Let $A$ be a metric space, and consider the $k$-dispersion problem in $A$ as in Definition~\ref{def:dispersion}. Consider an oracle $O$ that, given a set $P_{i}:= \{p_1,\cdots,p_i\}$ of $i$ points in $A$, returns a point $p \in A \setminus P_i$ such that $\sum_{j=1}^{i} d(p,p_{i}) \cdot a \geq \sum_{j=1}^{i} d(p',p_{i})$ for all $p' \in A \setminus P_i$ for some constant $a>1$. Then in $k$ calls to $O$ one can obtain a $2a$-approximate solution to the $k$-dispersion problem on $A$. In other words, the total pairwise distance between the $k$ points of the solution obtained via $O$ is at most that of the optimal solution to the $k$-dispersion divided by $2a$.
\end{lemma}
\begin{proof}
We will prove the statement of the lemma for the average distance, which will imply the claim for the total distance. Let $OPT$ be the average distance between  solutions in the optimal solution to the $k$-dispersion problem. By Lemma 1 in \cite{birnbaum09improved}, at any step $i$, there is always a solution whose distance to the already selected set $P_{i}$ is at least $OPT/2$, and hence the same holds for the farthest solution. This implies that the point $p$ returned by the algorithm has a distance at least $OPT/2a$. A simple inductive argument now finishes the proof: assume the lemma holds until stage $i$, i.e., the average pairwise distance between points in $P_{i}$ is at least $OPT/2a$. When $p_{i+1}=p$ is chosen, its average distance to the $i$ points in $P_{i}$ is at least $OPT/2a$ by the above reasoning, and so the $P_{i+1}$ has an average pairwise distance of at least $OPT/2a$. 
\end{proof}

\begin{lemma}\label{largermetric}
 Assume $A$ is a metric space, and $B \subseteq A$. Suppose there is an oracle $O$ that, given a set of $i$ points $P_{i}:=\{p_1,\cdots,p_{i}\}$ in $A$, outputs a point $p \in A \setminus P_{i}$ such that $\sum_{j=1}^{i} d(p_{j},p) \geq \max_{p' \in B} \sum_{j=1}^{i} d(p_{j},p')$. Then the oracle can be used to give a set of $k$ points in $A$ whose diversity is at least $1/4$ of an optimal solution to the $k$-dispersion problem in $B$. 
\end{lemma}
\begin{proof}
For a point $p \in A$, denote by $n_{B}(p)$ its closest point in $B$. By definition, if $p \in B$ then $n_{B}(p)=p$. 

Assume that the oracle $O$ in the hypothesis ofthe lemma exists. Then given the point set $P_{i}$, consider the point set $Q_{i} =\{ n_{B}(p_{j}) : 1 \leq j \leq i \}$. As mentioned before, a careful analysis of Lemmas 3-6 in \cite{birnbaum09improved} reveals that there always exists a point $p^{*} \in B$ whose average distance to the points in $Q_{i}$ is at least $OPT(B)/2$, where $OPT(B)$ denotes the optimal average pairwise distance between points in the optimal solution to the $k$ dispersion problem on $B$.

Consider the three points $p_{i}, n_{B}(p_{i})$ and $p^{*}$. By the triangle inequality, $d(p_{i},p^{*}) \geq d(p^{*}, n_{B}(p_{i})) - d(p_{i}, n_{B}(p_{i}))$. However, by definition, $d(p_{i}, n_{B}(p_{i})) \leq d(p_{i},p^{*})$. This implies that $d(p_{i},p^{*}) \geq d(p^{*}, n_{B}(p_{i})) - d(p_{i}, p^{*})$, implying that $d(p_{i},p^{*}) \geq d(p^{*}, n_{B}(p_{i}))/2$. Summing over all $i$, we get that the total (and hence the average) distance of $p^{*}$ to points $P_{i}$ is at least half of its total distance to points in $Q_{i}$. Since the average distance of $p^{*}$ to points in $Q_{i}$ is at least $OPT(B)/2$, we get that the average distance of $p^{*}$ to points in $P_{i}$ is at least $OPT(B)/4$. 

However, the oracle $O$ returns a point $p \in A$ whose average distance to $P_{i}$ is at least that of $p^{*}$ to points in $P_{i}$, which as argued above is at least $OPT(B)/4$. By the same inductive argument as in Lemma~\ref{farthestapprox}, we get that the average distance between points returned by $k$ calls to $O$ is at least $OPT(B)/4$, proving that it gives a $4$-approximation. 
\end{proof}

Given the above lemmas, we now prove the theorem.
\begin{itemize}
    \item Type 1: An $(a,1)$ bi-approximation to $BCO(\Pi, (y_1,\ldots,y_i), c, d)$ returns a point in $\Omega(c,(y_1,\ldots,y_i) )$. Therefore we only apply Lemma~\ref{farthestapprox} and we are done.
    \item Type 2: An $(a,b)$ bi-approximation to $BCO(\Pi, (y_1,\ldots,y_i), c, d)$ returns a point in $\Omega(bc,(y_1,\ldots,y_i) )$, that is an $a$-approximation to the farthest point in $\Omega(c, (y_1,\ldots,y_i) )$. Applying Lemma~\ref{farthestapprox} and Lemma~\ref{largermetric} together, we get the claimed result.
    \item Type 3: A $(1,1)$ bi-approximation to the similarity $BCO(\Pi, (y_1,\ldots,y_i), c, s)$ returns the farthest point in $\Omega(c, (y_1,\ldots,y_i))$, and so this case follows from the $2$-approximation guarantee of the farthest insertion heuristic, without applying either lemma.
    \item Type 4: A $(1,b)$ bi-approximation to $BCO(\Pi, (y_1,\ldots,y_i), c, d)$ returns the farthest point in $\Omega(bc,(y_1,\ldots,y_i) )$. In this case, we only apply Lemma~\ref{largermetric} to get the claimed result.
\end{itemize}

We prove the remaining case of Type 5 separately. A $(1+\epsilon, 1)$  bi-approximation to $BCO(\Pi, (y_1,\ldots,y_i), c, s)$ returns a point in $\Omega(c,(y_1,\ldots,y_i)) $. However, \textit{this point is not the farthest point, nor necessarily an approximation of it}. This is because the guarantee is only on the total similarity, and not the total distance. In other words, a $1+\epsilon$ approximation to $g_{s}(y):=\sum_{j=1}^{i} s(y,y_{i})$ is not a $1+\epsilon$ approximation to $f_{d}(y):= \sum_{j=1}^{i} d(y,y_{i})$-- in fact, these functions are the ``opposite'' of each other, as $f_{d}(y) = D i - g_{s}(y)$.

Let $y^{*}$ be the point in $\Omega(c, (y_1,\ldots,y_i))$ that maximizes $f_{d}(y)$, i.e., the farthest point. The $(1+\epsilon,1)$ bi-approximation algorithm returns a point $y$ such that $g_{s}(y) \leq a \cdot g_{s}(y^{*})$. Consider the condition that the average pairwise distance in the optimal solution $OPT$ to the DCP is at least $D \frac{4\epsilon}{1+2\epsilon}$. We refer the reader now to Lemma 1 in \cite{birnbaum09improved}, which proves that during farthest point insertion, at step $i+1$ in the algorithm when $P_{i} = \{p_1,\cdots,p_{i}\}$ have been selected, there is a point $p \notin P_{i}$ whose average distance is at least $OPT/2$. While this lemma seems like it can only be applied to point sets $P_{i}$ formed during an iterated farthest insertion, a careful analysis of Lemmas 3,4,5 and 6 in \cite{birnbaum09improved} used to prove it reveals that the argument does not demand that this be the case. That is, for any set $P_{i}$ of $i$ points, there always exist a $p \notin P_{i}$ whose average distance is at least $OPT/2$.

This implies the existence of a point $y' \in \Omega(c, (y_1,\ldots,y_i))$ such that the average distance of $y'$ to the $y_{i}$s is at least $OPT/2 \geq D \frac{2\epsilon}{1+2\epsilon} = D (1- \frac{1}{1+2\epsilon})$. Since $y^{*}$ is the farthest point, the same condition holds for $y^{*}$.

This implies that the total distance $f_{d}(y^{*}) \geq D i  (1- \frac{1}{1+2\epsilon})$, implying that $g_{s}(y^{*}) = Di -f_{d}(y^{*}) \leq Di/(1+2\epsilon)$. We then have the following string of inequalities

\begin{eqnarray}
g_{s}(y^{*}) &\leq& \frac{Di}{1+2\epsilon} \notag \\
\iff Di &\geq& (1+2\epsilon) g_{s}(y^{*}) \notag\\
\iff \frac{Di}{2} &\geq& (\frac{1}{2}+\epsilon)g_{s}(y^{*}) \notag \\
\iff Di - (1+\epsilon) g_{s}(y^{*}) &\geq& \frac{1}{2}(D - g_{s}(y^{*})) \notag 
\end{eqnarray}

But we know that the returned point $y$ satisfies $g_{s}(y) \leq (1+\epsilon) g_{s}(y^{*})$, implying that $f_{d}(y) = Di - g_{s}(y) \geq\frac{1}{2}(D - g_{s}(y^{*})) = f_{d}(y^{*})/2$. Hence $y$ is a $2$-approximation to the farthest point, and we can apply Lemma~\ref{farthestapprox} to complete the proof.
\end{proof}

A few remarks are in order:

\begin{itemize}
    \item The above theorem provides a recipe for solving the diversity computational problem for any given optimization problem. As long as \textit{either} the metric or the similarity budget constrained optimization problems can be solved or approximated in polynomial time, one has an analogous result for the DCP.
    \item In the remainder of this paper we will see several application that follow from the above 5 ``types'' of bi-approximations available. These include DCP for Maximum Matching (Type 1), DCP for shortest path (Type 3), DCP for minimum weight bases of a matroid, minimum spanning trees (Types 4 and 5).
    \item Whenever either $a$ or $b$ (or both) is set to be $1+\epsilon$, we call a bi-approximation for the BCO problem an FPTAS if the running time is polynomial in $1/\epsilon$ in addition to being polynomial in $d$ and $i$. Otherwise we call it a PTAS.
    
\end{itemize}

\noindent\textbf{Relation to Multicriteria Optimization:} Observe that for similarity BCOs, we need either $a$ or $b$ to be $1$. This class of biobjective problems that have a PTAS that is exact in one of the criteria is a special case of the multicriteria problems that have a PTAS that is exact in one of the criteria. Herzel et al. \cite{herzel2021one} showed that this class is exactly the class of problems for which the DUALRESTRICT version of the problem, posed by Diakonikolas and Yannakakis \cite{diakonikolas2010small}), can be solved in polynomial time. These are also the class of problems having a polynomial-time computable approximate $\epsilon$-Pareto set that is exact in one objective. This equivalence means that our theorem is applicable to this entire class of problems.

\subsection{Relaxed BCOs and Self-Avoidance} Before we delve into our applications, we describe another challenge in directly applying  results from multicriteria optimization literature. For a BCO, the second constraint demands that $y \in \SOL_{\Pi} \setminus \{y_1,\cdots,y_i\}$. Intuitively $y$ is the farthest point to the set of already discovered solutions $\{y_1,\cdots,y_i\}$, and because it is defined implicitly, without the second constraint $y$ may equal one of the $y_j$ ($1\leq j \leq i$). Consider an alternate BCO, which we call $BCO^{r}$ where the constraint is relaxed to $y \in \SOL_{\Pi}$.  For many graph problems, solving $BCO^{r}$ combined with the approach by Lawler \cite{Lawler72} gives a solution to the original BCO. This is extremely useful because most of the literature on multicriteria optimization concerns optimization of the relaxed type of problems $BCO^{r}$, and one can borrow results derived before without worrying about the second constraint. We remark that for other problems, $k$-best enumeration algorithms (see \cite{gabow1977two, hara2017enumerate,Lawler72,lindgren2017exact,murty1968algorithm}
for examples) may be useful to switch from the BCO to its relaxed version. \textit{Thus any algorithm for $BCO^{r}$ can be used, modulo the self-avoiding constraint (to be handled using Lawler's approach), to give a polynomial time algorithm for the Diversity Computational Problem with the same guarantees as in Theorem~\ref{thm:reduction}}. We provide examples of the approach by Lawler in subsequent sections where we consider specific problems. 

%% file: Application1-3.tex
\section{Application 1: Diverse Spanning Trees}\label{sec:st}

In this section, we discuss the diverse spanning trees problem, which is the diversity computational problem for spanning trees with Hamming distance function as the diversity measure. Let
	$G=(V,E)$ be an undirected graph. The problem aims to output a set $S$ of $k$ spanning trees $T_1,\cdots,T_k$ of $G$ such that the sum of the pairwise distances $\sum_{i,j \in S} d(T_i,T_j) $ is maximized, where $d$ is the Hamming distance between the edge sets of the trees. While this problem actually has an exact algorithm running in time $O((kn)^{2.5}m)$ \cite{hanaka2021finding}, we get a faster approximation algorithm.
	
		 \begin{theorem}\label{spanningtreetheorem}
Given a simple graph $G = (V, E)$, there exists an $O(knm\alpha(n, m))$-time algorithm, where $\alpha(\cdot)$ is the inverse of the Ackermann function,  that generates $k$ spanning trees $T_1,\cdots, T_k$, such that the sum of all pairwise Hamming distances is at least half of an optimal set of $k$ diverse spanning trees.
\end{theorem}
\begin{proof}
We develop an exact $(1,1)$ polynomial time subroutine for the associated BCO problem. 
Using our reduction Theorem~\ref{thm:reduction}, we observe that since one is only required to output spanning trees, the associated budget constrained problem has no inequalities. Having obtained $i$ trees $T_1,\cdots,T_i$ (starting with an arbitrary spanning tree $T_1$), the BCO problem looks like

\begin{equation}
\begin{aligned}
\max \quad & f_{d}(y):= \sum_{j=1}^{i} d(T,T_i)\\
\textrm{s.t.} \quad & T \in \SOL_{\Pi} \setminus \{T_1,\ldots,T_i\}  \\
\end{aligned}
\end{equation}

The relaxed budget constrained problem then simply asks to maximize $\sum_{j=1}^{i} d(T,T_i)$. We first show how to solve this problem exactly, and then adapt the approach by Lawler \cite{Lawler72} to handle the self-avoiding constraint. The algorithm to maximize $\sum_{j=1}^{i} d(T,T_i)$ is very simple: give each edge $e$ a weight $w(e)=\sum_{j=1}^i \1(e \in T_j)$ and compute the minimum spanning tree  $T$ with respect to these edge weights.

\begin{lemma}
\label{lem:opt_k-st}
    The minimum spanning tree $T$ with respect to the edge weights $w$ satisfies  $\sum_{j=1}^{i} d(T_{j},T) \geq \sum_{j=1}^{i} d(T_{j},T')$ for any spanning tree $T'$.
\end{lemma}
\begin{proof}
Given input spanning trees $T_1,\cdots,T_{i}$, suppose $T^A$ is the tree returned by the algorithm in Section~4.1, and $T^*$ is an optimal spanning tree. Let $E_j$ be the subset of edges with weight $j$. That is, $E_j=\{e \in E| w(e)=j\}$ for all $0 \leq j \leq i$ and partition $E$ as $E=E_0 \cup E_1 \cup \cdots \cup E_i$. 

Let $c_j = |T^A \cap E_j|$ and $\beta_j = |T^* \cap E_j|$. To prove that $T^A$ is optimal, we will show that $\sum_{j=1}^i d(T^A,T_j) \geq \sum_{j=1}^i d(T^
*,T_j)$. 

Define $A=\sum_{j=1}^i d(T^A,T_j)$ and $B=\sum_{j=1}^i d(T^
*,T_j)$. Plugging  $\alpha_j$ and $\beta_j$ into the sum of pairwise distances from our tree and the optimal tree to constructed trees, we get 
\[A=\sum_{j=1}^i d(T^A,T_j) = (n-1)i -\sum_{j=1}^i j\alpha_j\] and
\[ B=\sum_{j=1}^i d(T^
*,T_j) = (n-1)i -\sum_{j=1}^i j\beta_j.\]
We will show that $A -B =\sum_{j=0}^i j(\beta_j -\alpha_j) \geq 0.$ We first express $A-B$ as the dot product of two vectors

\[A-B = <0,1, \cdots, i> <\beta_0-\alpha_0,\beta_1-\alpha_1,\cdots,\beta_i-\alpha_i>^t,\] where $<.>^t$ denotes transpose, and makes a column vector out of a row vector. We develop some more notation:
\[E^j:=E_0 \cup E_1 \cup \cdots \cup E_j \]
\[\alpha^j:= \alpha_0+\alpha_1 +\cdots + \alpha_j\]
\[\beta^j:= \beta_0+\beta_1 +\cdots + \beta_j\]
\[c^j := <0,1, \cdots,j>\]
\[\alpha_v^j:=<\alpha_0,\alpha_1 ,\cdots , \alpha_j>^t\]
\[\beta_v^j:=<\beta_0,\beta_1 ,\cdots , \beta_j>^t\]
With this we have that $A-B =c^i(\beta_v^i-\alpha_v^i).$

We will show $c^i \beta_v^i \geq c^i \alpha_v^i$ by induction on $i$, which will prove that $A \geq B$, completing the proof. 

We first claim that $\alpha^j \geq \beta^j$ for all $0 \leq j \leq i.$ To prove this, fix a $j$ and let $C_1$ be the number  of connected components (including any isolated vertices) of $T^A \cap E^j$ and $C_2$ be the number  of connected components of $T^* \cap E^j$. Clearly, $C_1 \leq C_2$ since for Kruskal's algorithm for minimum spanning trees is the greedy algorithm that adds as many edges from $E^j$ into the solution as possible without creating a cycle, and the addition of an edge decreases the number of connected components by one. Since,  $C_1 = n - \alpha^j$ and $C_2 = n - \beta^j$, $C_1 \leq C_2$ implies that $\alpha^j \geq \beta^j.$ 

Now, we continue our proof that $c^i \beta_v^i \geq c^i \alpha_v^i$ by induction on $i.$ Observe $||\alpha_v^i||_1=||\beta_v^i||_1=n^{'} :=n-1,$ as both spanning trees contain exactly $n-1$ edges.

\noindent\textbf{Base Case} For the base case of $ i = 1$, $c^1=<0,1>, \alpha_v^1=<\alpha_0,n'- \alpha_0>,$  and $\beta_v^1=<\beta_0,n'- \beta_0>.$ From the previously mentioned property, we know $\alpha_0= \alpha^0 \geq \beta^0 = \beta_0.$ Hence, $c^1 \beta_v^1 = n'-\beta_0 \geq n-\alpha_0= c^1\alpha_v^1.$

\noindent\textbf{The induction hypothesis} assumes that $c^i \beta_v^i \geq c^i \alpha_v^i$  for $i = K-1$, i.e., the statement is true for two vectors of length $K-1$. Now, consider when $i=K$. Since $||\alpha^{K-1}_v||_{1} = \alpha^{K-1} \geq ||\beta^{K-1}_v||_{1}=\beta^{K-1}$, there are two cases,

\noindent\textbf{Case 1)}$: ||\beta^{K-1}_v||_{1}=||\alpha^{K-1}_v||_{1} $  Then since $||\beta^{K}_v||_{1}=||\alpha^{K}_v||_{1}$, we get that $\alpha_K =\beta_K$. Using the induction hypothesis, we get 
\[c^K \beta_v^K= c^{K-1} \beta_v^{K-1} + K \beta_K \geq c^{K-1} \alpha_v^{K-1} + K \alpha_K = c^K \alpha_v^K.\]

\noindent\textbf{Case 2)} $||\alpha^{K-1}_v||_{1} > ||\beta^{K-1}_v||_{1}$. Then  $||\alpha_v^K||_1=||\beta_v^K||_1$ yields $\alpha_K <\beta_K.$ Let $d = \beta_K- \alpha_K = ||\alpha^{K-1}_v||_{1} -||\beta^{K-1}_v||_{1}.$ Then, $c^K \beta_v^K - c^K \alpha_v^K=c^{K-1} (\beta_v^{K-1} -\alpha_v^{K-1})+Kd$. We claim that this quantity is positive. This is because the minimum value of $c^{K-1} (\beta_v^{K-1} -\alpha_v^{K-1})$, subject to the conditions that a) $||\alpha^{K-1}_v||_{1} - ||\beta^{K-1}_v||_{1}=d>0$ and b) $||\alpha^{K}_v||_{1} = ||\beta^{K}_v||_{1}$, occurs when $\beta_{j} =\alpha_{j}$ for all $0 \leq j \leq K-2$, $\beta_{K-1}=\alpha_{K-1} - d$, and $\beta_{K}=\alpha_{K}+d$. For this setting, the value of $c^{K-1} (\beta_v^{K-1} -\alpha_v^{K-1})+Kd$ equals $-(K-1)d+Kd=d$, which is positive, completing the proof.
\end{proof}

If $T \neq T_j$ for all $1 \leq j \leq i$, that is, the new spanning tree is different from all previous trees, then we are done and Theorem~\ref{spanningtreetheorem} is proved using Lemma~\ref{lem:opt_k-st}.

\noindent\textbf{Self-Avoiding Constraint}  However, this is not guaranteed as our measure is the sum-of-distances and not the minimum distance. Note that for furthest point insertion, we need the point that is furthest from the current set, but does not belong to the current set. This is an issue we face because of the implicit nature of the furthest point procedure in the exponential-sized space of spanning trees: in the metric k-dispersion problem, it was easy to guarantee distinctness as one only considered the $n-i$ points not yet selected. 

We now show how to guarantee that the new tree is distinct. In case that $T = T_j$ for some $1 \leq j \leq i$, we use the approach by Lawler~\cite{Lawler72}. We will obtain $i+1$ distinct spanning trees $\widetilde{T}_1, \cdots, \widetilde{T}_{i+1}$, at least one of which must then be distinct from $T_1, \cdots,T_i$, which will be chosen as $T_{i+1}$. 
Set $\widetilde{T}_{1}=T_{j}$.
For every edge $e \in \widetilde{T}_{1}$, we find the minimum spanning tree $T_e$ (with respect to the same edge weights $w(e)$ as before) after deleting $e$ from the graph.  Among the collection $\{T_{e}\}_{e \in \widetilde{T}_{1}}$ of trees thus obtained, we find the one whose sum of distances from $T_1, \cdots,T_i$ is as large as possible, and set it to be $\widetilde{T}_{2}$. Note that  $\widetilde{T}_{2} \neq \widetilde{T}_{1}$. If $\widetilde{T}_{2} \neq T_{j}$ for all $1 \leq j \leq i$, then we are done; else,  we repeat the above procedure to obtain the collection $\{T_e\}_{e \in \widetilde{T_2}}$ and set $\widetilde{T}_{3}$ as the best one in $\{T_e\}_{e \in \widetilde{T_1}} \cup \{T_e\}_{e \in \widetilde{T_2}} \setminus \{\widetilde{T}_{1}, \widetilde{T}_{2}\}$, after which we are either done or we repeat to get  $\widetilde{T}_{4}$, and so on. We stop after obtaining  $\widetilde{T}_{i+1}$, and select the $\widetilde{T}_{(.)}$ that is distinct from all of $T_1, \cdots,T_i$. The proof of Lemma~\ref{lem:opt_k-st} implies that the $T_{i+1}$ thus obtained is the furthest tree from the $T_1, \cdots,T_i$ among all trees $T \notin \{T_1, \cdots,T_i\}$, and we have accomplished furthest point insertion in the space of spanning trees. 
In total, the above needs to compute $O(kn)$ instances of the minimum spanning tree, so the running time is $O(k n m \alpha(n, m))$~\cite{Chazelle00}.
\end{proof}

\section{Application 2: Diverse Approximate Shortest Paths}\label{sec:shortestpath}

Given a graph $G=(V, E)$, non-negative edge weights $w(e)$, two vertices $s$ and $t$, and a factor $c>1$, the diversity computational problem asks to output $k$ many $st$ paths, such that the weight of each path is within a factor $c$ of the weight of the shortest $st$ path, and subject to this constraint, the total pairwise distance between the paths is maximized. Here the distance between two paths is again the Hamming distance, or size of symmetric difference of their edge sets.

In~\cite{Hanaka21}, it is shown that finding $k$ shortest  paths with the maximum diversity (i.e. the average Hamming distance between  solutions) can be solved in polynomial time, but finding $k$ ``short'' paths with the maximum diversity is NP-hard. In contrast, in what follows, we will show that finding $k$ ``short'' paths with constant approximate diversity is  polynomial-time solvable.

\textit{We will show that the associated budget constrained optimization problem for this is of Type 3 in Theorem~\ref{thm:reduction}}. In other words, we will show that the BCO can be solved exactly. This will result in a $(2,c)$ approximation algorithm for the diversity computational problem.

Hence, we need an algorithm that implements: given a set $S$ of $c$-$st$-shortest paths, find a $c$-$st$-shortest path $P \notin S$ so that $W(S \cup \{P\})$ is maximum among all $W(S \cup \{P'\})$ for $c$-$st$-shortest path $P' \notin S$. Here, $W(
S')$ is the sum of all pairwise Hamming distances between two elements in $S'$. This is a special case of the bicriteria shortest paths, for which there is an algorithm in~\cite{NC82}. In our case, \textbf{one of the two weight functions is an integral function with range bounded in $[0, k]$}. Hence, it can be solved more efficiently than the solution in~\cite{NC82}, which can be summarized as following.

\begin{lemma}[Exact solution to the relaxed $BCO^{r}$ problem]
\label{lem:BiSP}
Given a real $c \ge 1$ and a directed simple graph $G = (V \cup \{s, t\}, E)$ associated with two weight functions on edges $\omega: E \rightarrow \mathbb{R}^+$ and $ f: E \rightarrow \{0, 1, \ldots, r\}$,  there is an $O(r|V|^3)$-time algorithm to output an $st$-path $P^*$ so that  
$\sum_{e \in E(P^*)} f(e)$ is minimized while retaining $\sum_{e \in E(P^*)} \omega(e) \le c \sum_{e \in E(P)} \omega(e)$ for all $st$-paths $P$.
\end{lemma}
\begin{proof}
Let $D(.,.)$ be a $|V| \times r(|V|-1)$ table with the entry values defined as the following:
\[
D(v, c) = \min\left\{\sum_{e \in E(P)} \omega(e) : P \mbox{ is an } sv\mbox{-path with} \sum_{e \in E(P)} f(e) \le c\right\}.
\]
If the RHS takes the minimum from an empty set, then we define $D(v, c) = \infty$. 
By definition, $P^*$ has $\sum_{e \in E(P^*)} f(e) = c^*$ so that \[ 
    c^* = \min\{c \in [0, r(|V|-1)] : D(t, c) \le c \cdot dis_\omega(s, t)\}
\]
where $dis_\omega(s, t)$ denotes the distance under $\omega$ from $s$ to $t$.
By backtracking from the $D(t, c^*)$, one can construct $P^*$ from the table $D$ in $O(r|V|)$ time. Hence, our goal is to fill out the entries in $D$ using $O(r|V|^3)$ time.

To compute $D(v, 0)$ for all $v \in V$, it suffices to run Dijkstra's algorithm~\cite{Dijkstra59} on the graph obtained from $G$ with the removal of the edges whose $f(e) > 0$. This step takes $O(|V|^2)$ time. 

Given $D(v, c')$ for all $v \in V, c' \in [0, c-1]$, to compute $D(v, c)$ for all $v \in V$ there are two cases to discuss. We say an $sv$-path $P$ \emph{admits} $D(v, c')$ if $P$ has $\sum_{e \in E(P)} \omega(e) = D(v, c')$ and $\sum_{e \in E(P)} f(e) \le c'$.

\begin{enumerate}
    \item[] Case 1. Some $sv$-path $(s, \ldots, x, v)$ with $f((x, v)) > 0$ admits $D(v, c)$.
    \item[] Case 2. Some $sw$-path $P = (s, \ldots, y, w)$ with $f((y, w)) > 0$ admits $D(w, c)$, and $P$ together with some edges $e$ whose $f(e) = 0$ admits $D(v, c)$.
\end{enumerate}

For Case 1, it suffices to check $D(x, c-f((x, v)))$ for all the neighbors of $v$ with $f((x, v)) > 0$. This step takes $O(|V|)$ time. For Case 2, observe that $D(w, c)$ can be determined by Case 1. It suffices to initialize $D(v, c)$ for all $v \in V$ as the values obtained from Case 1 and update $D(v, c)$ if there is a path $Q$ from some node $w$ to $v$ so that every edge $e \in E(Q)$ has $f(e) = 0$ and $D(w, c) + \sum_{e \in E(Q)} \omega(e) < D(v, c)$. This can be implemented in $O(|V|^2)$ time for each $c$ by running Dijkstra's algorithm on the graph $G = (V, E \cup E_{1} \setminus E_2)$ where $E_1$ is a set of new directed edges $(s, v)$ for all $v \in V$ with $f((s, v)) = 0$ and $\omega((s, v))$ as the initial $D(v, c)$ obtained from Case 1, and $E_2$ is the set of all the edges $e \in E$ that have $f(e) > 0$. It is worth noting that, though we do not explicitly check whether the found paths contain cycles or not (i.e. not simple paths), the paths found to admit $D(v, c)$ must be simple since all shortest paths in a positive weight graph are simple.

To sum up, an optimal solution $P^*$ can be found in $O(r|V|^3)$ time. 
\end{proof}

\noindent\textbf{Self-Avoiding Constraint} We now turn to solving the associated (non-relaxed) $BCO$ problem, by generalizing the above lemma to Corollary~\ref{cor:RBiSP}. Thus \cref{cor:RBiSP} will help us avoid the situation that a furthest insertion returns a path that is already picked by some previous furthest insertion.

\begin{corollary}[Exact solution to the $BCO$ problem]
\label{cor:RBiSP}
Given a real $c \ge 1$, a directed simple graph $G = (V \cup \{s, t\}, E)$ associated with two weight functions on edges $\omega: E \rightarrow \mathbb{R}^{+}$, $ f: E \rightarrow \{0, 1, \ldots, r\}$, and two disjoint subsets of edges $E_{in}, E_{ex} \subseteq E$ so that all edges in $E_{in}$ together form a directed simple path $P_{\rm prefix}$ starting from node $s$, there exists an $O(r|V|^3)$-time algorithm to output an $c$-$st$-shortest path $P^*$ under $\omega$ so that $\sum_{e \in E(P^*)} f(e)$ is minimum among all the $c$-$st$-shortest paths $P$ that contain $P_{\rm prefix}$ as a prefix and contain no edges from $E_{ex}$, if such an $c$-$st$-shortest path exists. 
\end{corollary}
\begin{proof}
Let $P_{\rm prefix} = (s, \ldots, v)$. This suffices to find a feasible $vt$-path $P_{\rm suffix} = (v, \ldots, t)$ on the graph obtained from $G$ with the removal of all the nodes in $P_{\rm prefix}$ except $v$ and with the removal of all edges in $E_{ex}$ and those incident to the deleted nodes. This can be computed by Lemma~\ref{lem:BiSP}.
\end{proof}

We are ready to state our main result for the diverse $c$-$st$-shortest paths.

\begin{theorem}[$(2,c)$ bi-approximation to the Diversity Problem on Shortest Paths]
\label{thm:MainSP}
For any directed simple graph $G = (V \cup \{s, t\}, E)$, given a $c>1$ and a $k \in \mathbb{N}$, there exists an $O(k^3 |V|^4)$-time algorithm that, if $G$ contains at least $k$ distinct $c$-$st$-shortest paths, computes a set $S$ of $k$ distinct $c$-$st$-shortest paths so that the sum of all pairwise Hamming distances between two paths in $S$ is at least one half of the maximum possible; otherwise, reports ``Non-existent.''
\end{theorem}
\begin{proof} Let $P_1$ be an arbitrary $1$-$st$-shortest path (i.e. a shortest path from $s$ to $t$), which can be computed in $O(|V|^2)$ time by Dijkstra's algorithm. We perform the furthest insertion mentioned in Lemma~\ref{4approx} to obtain the $c$-$st$-shortest paths $P_2, P_3, \ldots, P_k$. The collection $\{P_i : i \in [1, k]\}$ is a desired solution. 

To perform the $i$-th furthest insertion for $i \in [1, k-1]$, we need an algorithm for the problem that, given $\{P_j : 1 \le j \le i \}$, find an $c$-$st$-shortest path $P_{i+1}$ so that the following conditions both hold:

\begin{enumerate}
  \item $\sum_{e \in E(P_{i+1})} f_i(e)$ is minimum among all $c$-$st$-shortest paths where $f_i(e)$ for $e \in E$ is defined as the number of times that $e$ appears in $P_j$ for all $j \in [1, i]$. Hence, $f_i$ maps $E$ to $\{0, 1, \ldots, i\}$.
  \item $P_{i+1} \ne P_j$ for all $j \in [1, i]$.
\end{enumerate}

To find an $c$-$st$-shortest path that satisfies the first condition, it is an instance of the bicriteria shortest path problem with two weight functions $\omega$ and $f_i$. By Lemma~\ref{lem:BiSP}, this can be done in $O(i|V|^3)$ time. 

However, the $c$-$st$-shortest path found by minimizing $\sum_{e} f_i(e)$ is not necessarily different from those found in the previous furthest insertions. To remedy, one can find the $c$-$st$-shortest paths whose $\sum_{e} f_i(e)$ are the $j$-th smallest for all $j \in [1, i+1]$. Some of the $i+1$ $c$-$st$-shortest paths can meet the second condition. This step can be implemented by the standard approach due to Lawler~\cite{Lawler72}. That is, suppose in time $T$ one can compute $P_{i+1}$ with an additional constraint that some given edge set $E_{i, in}$ (all edges in $E_{i, in}$ together form a directed path starting from node $s$) must be contained in $E(P_{i+1})$ and some given edge set $E_{i, ex}$ must not be contained in $E(P_{i+1})$, then enumerating the $i+1$ smallest candidates for $P_{i+1}$ can be done in $O(i |V| T)$ time. By Corollary~\ref{cor:RBiSP}, $T = O(i|V|^3)$. Thus, the $i$-th furthest insertion can be done in $O(i^2 |V|^4)$ time. In total, our approach needs $O(k^3 |V|^4)$ time  as desired. 
\end{proof}

\section{Application 3: Diverse Approximate Maximum Matchings}\label{sec:matching}

Consider the diversity computational problem for computing $k$ many $c$ maximum matchings for undirected graphs. 
In~\cite{fomin2020diverse}, the authors present an algorithm, among others, to find a pair of maximum matchings for bipartite graphs whose Hamming distance is maximized. In contrast, our result can be used to find $k \ge 2$ approximate maximum matchings for any graph whose diversity (i.e. the average Hamming distance) approximates the largest possible by a factor of $2$.

Since $\goal_{\Pi} = \max$, using the Hamming metric on the edge set of the solutions gives us a BCO with minimization in the objective function (the sum of the total distances) and constraints with lower bound.

We show that this problem can be restated into the budgeted matching problem~\cite{BergerBGS11}. As noted in~\cite{BergerBGS11}, though the budgeted matching is in general $\mathcal{NP}$-hard, if both the weight and cost functions are integral and have a range bounded by a polynomial in $|V|$, then it can be solved in polynomial time with a good probability by a reduction to the exact perfect matching problem~\cite{CameriniGM92,MulmuleyVV87}. The exact running time for such a case is not stated explicitly in~\cite{BergerBGS11}. We combine the algorithm in~\cite{BergerBGS11} and the approach by Lawler \cite{Lawler72} to prove our main theorem for diverse matchings (\cref{thm:matchmain}). 

Define the distance between two matchings as the Hamming distance between the edge sets. The associated relaxed $BCO^{r}$ is a maximizatin problem with lower bounded constraints. It can be restated as: given a non-empty set $S$ of $c$-maximum matchings, find a $c$-maximum matching $M \notin S$ so that $W(S \cup \{M\})$ is maximum among all $W(S \cup \{M'\})$ for $c$-maximum matching $M' \notin S$ (here $W(X)$ is the sum of pairwise distances between all matchings in $X$. This problem can be restated into the budgeted matching problem~\cite{BergerBGS11}. As noted in~\cite{BergerBGS11}, though the budgeted matching is in general $\mathcal{NP}$-hard, if both the weight and cost functions are integral and have a range bounded by a polynomial in $|V|$, then it can be solved in polynomial time with a good probability by a reduction to the exact perfect matching problem~\cite{CameriniGM92,MulmuleyVV87}. The exact running time for such a case is not stated explicitly in~\cite{BergerBGS11}, and we analyze it as below.

\begin{lemma}[Restricted Budgeted Matching Problem]\label{lem:RBMP}
Given an undirected simple graph $G = (V, E)$ and a cost function $c: E \rightarrow \{0, 1, \ldots, r\}$ on edges, there exists an $O(r^2 |V|^6 \log^2 r|V|)$-time randomized algorithm that can find a matching of smallest cost among all $c$-maximum matchings with some constant success probability where the cost of a matching is the sum of the costs of all edges in the matching. 
\end{lemma}
\begin{proof}

Suppose that, for all $x \in [0, X]$ with $X = O(|V|)$, for all $y \in [0, Y]$ with $Y = O(r |V|)$, we can decide whether $G$ contains a matching $M$ that consists of $x$ edges and has cost $y$. Then, the desired matching $M^*$ can be found.

As the reduction in~\cite{BergerBGS11}, we construct an edge-weighted graph $H = (U, F)$ so that 
\begin{itemize}
    \item $U = V \cup \{z_1, z_2, \ldots, z_{|V|}\}$ where $z_i \notin V$ for all $i \in [1, |V|]$,
    \item $F = E \cup \{\{z_i, z_j\} : i \ne j \in [1, |V|]\} \cup \{\{x, z_i\} : x \in V, i \in [1, |V|]\}$, and
    \item for every edge $e \in F$, if $e$ is also in $E$, then it has weight $\Gamma+c(e)$ for some sufficiently large integer $\Gamma$ to be determined later; otherwise, it has weight $0$. 
\end{itemize}

By setting $\Gamma \ge r\lceil |V|/2\rceil$+1, $G$ has a matching $M$ that consists of $x$ edges and has cost $y$ if and only if $H$ has a perfect matching of weight $x\Gamma + y$. This is an instance of the exact perfect matching problem~\cite{CameriniGM92,MulmuleyVV87}, which can be solved by a randomized algorithm in $O(r|V|^4 \log r|V|)$ time with some constant success probability. Summing over all choices of $x$ and $y$, the running time is $O(r^2 |V|^6 \log^2 r|V|)$.
\end{proof}

We are ready to prove our main result for the diverse $c$-maximum matchings.

\begin{theorem}\label{thm:matchmain}
There exists a $O(k^{4}|V|^{7} \log^{3} k|V|)$ time, $(2,c)$ bi-approximation to the diversity computational problem for $c$-maximum matchings, with failure probability $1/ |V|^{\Omega(1)}$.
\end{theorem}
\begin{proof}
 Let $M_1$ be an arbitrary maximum matching, which can be computed in $O(|V|^{2.5})$ time~\cite{MicaliV80}. For each $i \in [1, k-1]$, given $S_i = \{M_j : j \in [1, i]\}$, find a $c$-maximum matching $M_{i+1} \notin S_i$ so that $W(S_i \cup M_{i+1})$ is maximum among all $W(S_i \cup M')$ for $c$-maximum matching $M' \notin S_i$. By Lemma~\ref{lem:RBMP}, set the cost function $c_i(e)$ for every $e \in E$ as the number of times that edge $e$ appears in $M_j$ for all $j \in [1, i]$ and apply the algorithm for Lemma~\ref{lem:RBMP} to solve the instance $(G = (V, E), c_i, c)$. This returns an $c$-maximum matching $M^\dagger$ so that $W(S_i \cup M^\dagger)$ is maximum among all $W(S_i \cup M')$ that $M'$ is an $c$-maximum matching and may or may not be contained in $S_{i}$. Thus we need to guarantee self-avoidance.
 
 To remedy, we enumerate the best $i+1$ candidates for $M_{i+1}$. Some of them is not contained in $S_{i}$. This enumeration can be done by running the algorithm for Lemma~\ref{lem:RBMP} $O(i|V|)$ times, each of which preselects some edges to be included in and some to be excluded from  $M_{i+1}$, as the Lawler's approach~\cite{Lawler72}. For each $i \in [1, k-1]$, this step takes $O(k^3 |V|^7 \log^3 k|V|)$ time and may fail with probability $(k|V|)^{-\Omega(1)}$. The failure probability is $(k|V|)^{-\Omega(1)}$ rather than a constant because we can run the algorithm that may err $O(\log k|V|)$ times.

Summing up the running time for the $k-1$ furthest insertions, the total running time is $O(k^4 |V|^7 \log^3 k|V|)$, and the failure probability is $1/|V|^{\Omega(1)}$ by the Union bound.
\end{proof}

%% file: Application4.tex
\section{Application 4: Diverse Minimum Weight Matroid Bases and Minimum Spanning Trees}\label{sec:matroid}

One of the original ways to attack the peripatetic salesman problem (Krarup \cite{krarup1995peripatetic}) was to study the $k$ edge-disjoint spanning trees problem \cite{clausen1980finding}. Note that the existence of such trees is not guaranteed, and one can use our results in Section~\ref{sec:st} to maximize diversity of the $k$ trees found. 

However, for an application to the TSP problem, cost conditions must be taken into account. Here we study the diverse computational problem (DCP) on minimum spanning trees: Given a weighted undirected graph $G=(V,E)$ with nonnegative weights $w(e)$, $c>1$ and a $k \in \mathbb{N}$, return $k$ spanning trees of $G$ such that each spanning tree is a $c$-approximate minimum spanning tree, and subject to this, the diversity of the $k$ trees is maximized. Here again the diversity of a set of trees is the sum of pairwise distances between them, and the distance between  two trees is the size of their symmetric difference.

Our results in this section generalize to the problem of finding $k$ diverse bases of a matroid such that every basis in the solution set is a $c$ approximate minimum-weight basis. The DCP on MSTs is a special case of this problem. However, in order to not introduce extra notation and definitions here, we will describe our method for minimum spanning trees. We will then briefly sketch how to extend the algorithm to the general matroid case.

Starting with $T_{1}=MST(G)$ (a minimum spanning tree on $G$, computable in polynomial time), assume we have obtained $i$ trees $T_1,\cdots,T_i$, all of which are $c$-approximate minimum spanning trees. Assign to each edge a length $\ell(e)$ which equals $|j: 1 \leq j \leq i, e \in T_{j} |$. 

\begin{claim}\label{claimmst}
Given $T_1,\cdots,T_i$, finding $T_{i+1}$ that maximizes $\sum_{j=1}^{i} d(T,T_j)$ is equivalent to finding $T$ that minimizes $\sum_{e \in T} \ell(e)$.
\end{claim}

\begin{proof}
An explicit calculation reveals that $\sum_{e \in T} \ell(e) = (n-1)i - \sum_{j=1}^{i} d(T,T_j)$.
\end{proof}

Consider now the associated \textit{similarity budget constrained optimization problem}

\begin{equation}
\begin{aligned}
\min \quad & \sum_{e \in T} \ell(e)\\
\textrm{s.t.} \quad & w(T)  \leq c \cdot w(MST(G)) \\
  & T \in \SOL_{\Pi} \setminus \{T_1,\ldots,T_i\}    \\
\end{aligned}
\end{equation}

Here $\SOL_{\Pi}$ is just the set of spanning trees on $G$. We will handle the self-avoiding constraints in a similar fashion as in Section~\ref{sec:st}. For the moment consider the relaxed $BCO^{r}$ where the last constraint is simply $T \in \SOL_{\Pi}$. This is a budget constrained MST with two weights. This problem has been considered by Ravi and Goemans \cite{ravi1996constrained}, who termed it the CMST problem. They provide a $(1,2)$ bi-approximation that runs in near-linear time, and a $(1,1+\epsilon)$ bi-approximation that runs in polynomial time\footnote{The latter is a PTAS, not an FPTAS.}. Also, they show that the $(1,1+\epsilon)$ bi-approximation can be used as a subroutine to compute a $(1+\epsilon,1)$ bi-approximation in pseudopolynomial time.

Applying their results and observing that we are in cases 4 and 5 of Theorem~\ref{thm:reduction}, we get

\begin{theorem}[DCP for Mininum Spanning Trees] 
There exists a
\begin{itemize}
    \item polynomial (in $n,m$ and $k$) time algorithm that outputs a $(4,2c)$ bi-approximation to the DCP problem for MSTs.
    \item polynomial (in $n,m$ and $k$) and exponential in $1/\epsilon$ time algorithm that outputs a $(4,(1+\epsilon)c)$ bi-approximation to the DCP problem for MSTs.
    \item pseudopolynomial time algorithm that outputs a $(4,c)$ bi-approximation to the DCP problem for MSTs, as long as the average distance between the trees in the optimal solution to the $k$ DCP on $c$-approximate minimum spanning trees does not exceed $\frac{4\epsilon(n-1)}{1+2\epsilon}$.
\end{itemize}
\end{theorem}

\noindent\textbf{Extension to Matroids:} It is stated in the paper by Ravi and Goemans \cite{ravi1996constrained} that the same result holds if one replaces the set of spanning trees by the bases of any matroid. It is straightforward to show that the analog of Claim~\ref{claimmst} hold in the matroid setting too. With a bit of work, one can also generalize the approach of Lawler \cite{Lawler72} to avoid self-intersection (the bases found so far), and thus all the techniques generalize to the matroid setting. In all of this, we assume an independence oracle for the matroid, as is standard. In~\cite{fomin2021diverse}, it is shown that, given integers $k, d$, finding $k$ \textit{perfect} matchings so that every pair of the found matchings have Hamming distance at least $d$ is NP-hard. This hardness result also applies to finding weighted diverse bases and weighted diverse common independent sets.